\pdfoutput=1
\documentclass[a4paper,11pt]{article}
\usepackage[
top=30truemm,
bottom=30truemm,
left=25truemm,
right=25truemm
]{geometry} 
\usepackage[dvips]{graphicx}
\usepackage{amsmath,amsthm,amssymb}
\usepackage{amsfonts}
\usepackage{color}
\usepackage{pifont}
\usepackage{here}
\usepackage{cases}
\usepackage{empheq}
\usepackage{multirow} 
\usepackage{colortbl} 
\usepackage{arydshln} 
\usepackage{url} 
\usepackage[ 
algo2e,
boxed,
vlined 
]{algorithm2e} 
\SetAlgoSkip{medskip} 
\SetArgSty{textrm} 
\IncMargin{1ex} 
\SetProcNameSty{textsc} 
\SetKwRepeat{Do}{do}{while}

\def\dfrac{\displaystyle\frac}

\newtheorem{thm}{Theorem}[section]%
\newtheorem{df}{Definition}[section]%
\newtheorem{lem}{Lemma}[section]%
\newtheorem{rem}{Remark}[section]%
\begin{document}
\title{Shipper collaboration matching: \\ fast enumeration of triangular transports \\ with high cooperation effects}
\author{Akifumi~Kira $^{\tt a, *}$, Nobuo Terajima $^{\tt b}$, \\ Yasuhiko Watanabe $^{\tt b}$,  
and Hirotaka Yamamoto $^{\tt c}$ \\
}
\date{\empty}
\maketitle
\vspace{-10mm}
\begin{center}
\small
\begin{tabular}{l}
$^{\tt a}$ Faculty of Informatics, Gunma University \\
\quad 4-2 Aramaki-machi, Maebashi, Gunma 371-8510, Japan \\
$^{\tt b}$ Japan Pallet Rental Corporation \\
\quad Ote Center Building, 1-1-3 Otemachi, Chiyoda-ku, Tokyo 100-0004, Japan \\
$^{\tt c}$ Thincess Co. Ltd. \\
\quad Saito Building \#3A, 3-1-4 Yotsuya, Shinjuku-ku, Tokyo 160-0004, Japan \\
$^{\tt *}$ Corresponding author, E-mail: a-kira@gunma-u.ac.jp
\end{tabular}
\end{center}
\begin{abstract}
The logistics industry in Japan is facing a severe shortage of labor. 
Therefore, there is an increasing need for joint transportation allowing large amounts of cargo to be transported using fewer trucks. 
In recent years, the use of artificial intelligence and other new technologies has gained wide attention for improving matching efficiency. 
However, it is difficult to develop a system 
that can instantly respond to requests because browsing through enormous combinations of 
two transport lanes is time consuming. 
In this study, we focus on a form of joint transportation called triangular transportation 
and enumerate the combinations with high cooperation effects. 
The proposed algorithm makes good use of hidden inequalities, such as the distance axiom, 
to narrow down the search range without sacrificing accuracy. 
Numerical experiments show that the proposed algorithm 
is thousands of times faster than simple brute force.
With this technology as the core engine, we developed a joint transportation matching system.
The system has already been in use by over 150 companies as of October 2022, 
and was featured in a collection of logistics digital transformation cases 
published by Japan's Ministry of Land, Infrastructure, Transport and Tourism. 
\end{abstract}

\section{Introduction}
{\bf Background.} 
Owing to the severity of the recent logistics crisis and shortage of truck drivers in Japan, 
improving labor productivity has become an urgent issue. 
However, the load factor (truck fill rate) remains low, below 40\%~\cite{policy}. 
In other words, on average, trucks are loaded to only 40\% of their capacity. 
One reason for this is vacant return trips during long-haul transportation. 
Therefore, the Japan Pallet Rental Corporation (JPR) has provided significant support  toward joint transport 
by companies in other industries, as well as other measures to carry more cargo with fewer trucks. In October 2019, JPR recognized the importance of creating a system for deploying such initiatives throughout the logistics industry and 
began developing a common transportation matching system using artificial intelligence (AI) technology in collaboration with Gunma University. From October 2019 to March 2021, this development was funded by the New Energy and Industrial Technology Development Organization (NEDO)\footnote{
National Research and Development Agency under the jurisdiction of Japan's Ministry of Economy, Trade and Industry.
}
 \lq\lq The Project to Promote Data-Sharing in Collaborative Areas and Developing of AI Systems to Promote Connected Industries).''

\par 
{\bf Technical issues.} 
A series of transportation in which a single truck sequentially handles two transport lanes, 
such as from Tokyo to Osaka and from Osaka to Tokyo, so that the empty backhauls are reduced, 
is called \lq\lq round-trip transport." 
Furthermore, this is called "triangular transport" when it is expanded into three transport lanes, 
i.e., from Tokyo to Kanazawa, Kanazawa to Osaka, and Osaka to Tokyo, in the form of a triangle. 
The higher the occupied vehicle rate (i.e., the lower the empty running rate) is, 
the more efficient the process is (see Figure~\ref{fig:triangular}). 
In terms of negotiation efforts, 
round-trip transportation is ideal;   
however, there are no convenient partners with transport lanes in opposite directions, 
particularly for transportation between regional cities.  
In such cases, triangular transportation significantly increases the number of options available.
If Company A requests matching, the system searches for partners (Companies B and C) in its database and presents a list of efficient triangular transports. 
If 10,000 transport lanes are registered, there are approximately 100 million combinations of two transport lanes. 
Additionally, because it is necessary to calculate the distance traveled, which differs depending on the combination, a simple brute-force calculation would require a long time using a calculator. In reality, one logistics company has multiple transport lanes. 
Therefore, it is difficult for conventional methods to respond instantly to multiple matching requests from a single user or a series of requests from multiple users.
\begin{figure}[h]
\centering
\begin{tabular}{l}
$\longrightarrow$ : Loading Trip \\
$\dashrightarrow$\, : Empty Trip
\end{tabular}
\hspace{5mm}
$\textrm{Occupied Vehicle Rate} = \dfrac{\textrm{Length of} \longrightarrow}{\textrm{Length of} \longrightarrow + 
\dashrightarrow}$ \\[3mm]
\includegraphics[scale=0.14]{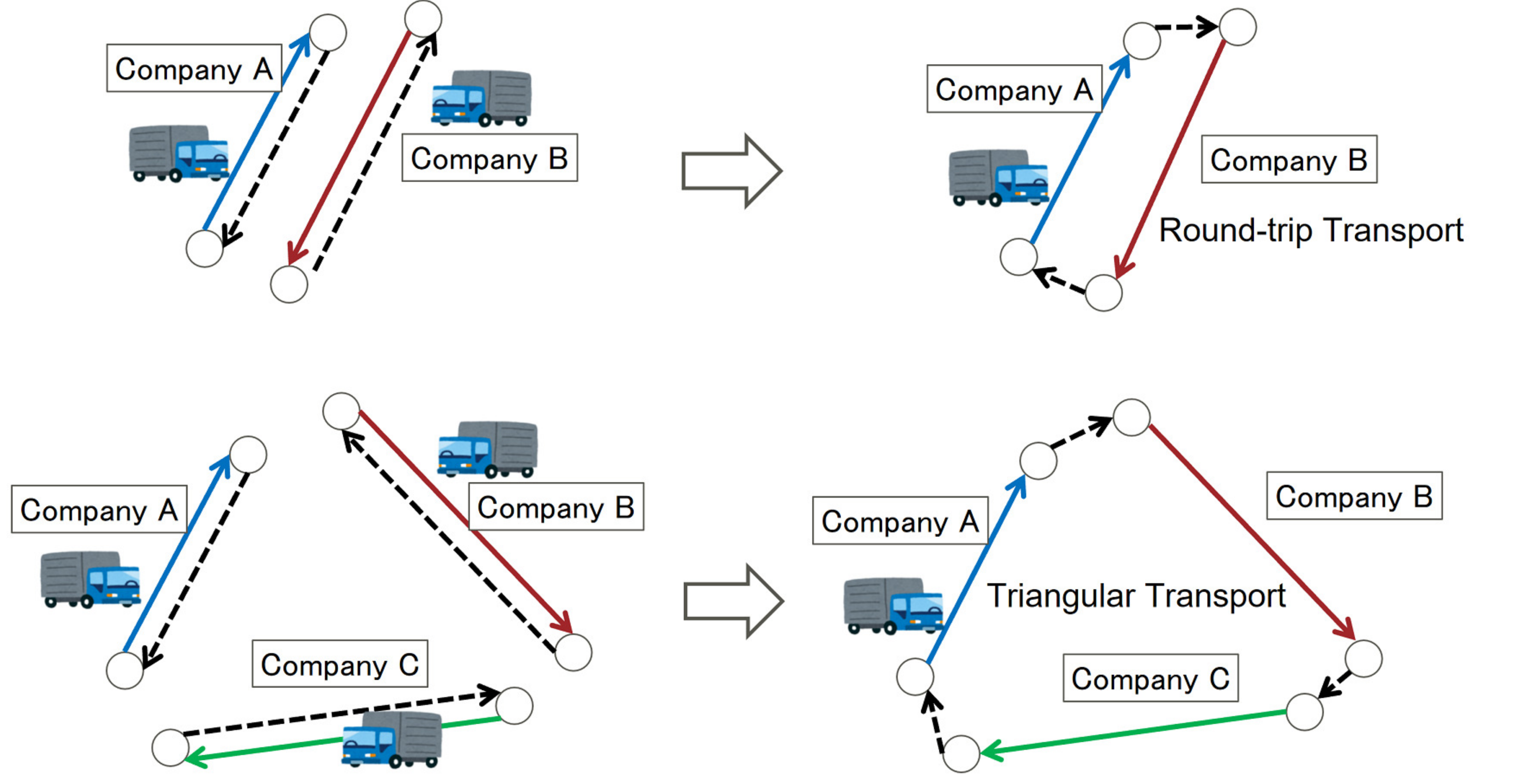}
\caption{Round-trip transportation and triangular transportation}
\label{fig:triangular}
\end{figure}
\par
{\bf Our contribution.} 
Suppose Company A requests a triangular transportation matching. Because, the first lane is that of Company A, we know the starting point, endpoint, and transportation distance of the first lane. The second and third lanes have not yet been finalized 
as we are searching for partners. Therefore, we take advantage of the underlying inequalities. 
Using distance axioms skillfully, we can back-calculate the conditions required to achieve the specified occupied vehicle rate. 
For instance, \lq\lq the vacant distance between the first and second lanes must be less than or equal to this value," and 
\lq\lq the distance of the second lane must be greater than or equal to this value." Therefore, the scope of the search is significantly reduced. By further modifying the data structure and traversal order, faster enumeration becomes possible
\footnote{Patent rights jointly filed by Gunma University and JPR on October 20, 2021~\cite{Kira2021}.}.
\par
{\bf Related Work.} 
Collaborative transportation has been extensively studied 
and related reviews by 
Cruijssen et al.~\cite{Cruijssen2007}, 
Verdonck et al.~\cite{Verdonck2013}, 
Guajardo and R\"{o}nnqvist~\cite{Guajardo2016}, 
Gansterer and Hart~\cite{Gansterer2018}, 
Pan et al.~\cite{Pan2019}, 
Karam et al.~\cite{Karam2021}, 
and 
Mrabti et al.~\cite{Mrabti2022} 
exist. 
The enumeration problems addressed in this study 
are related to the type of horizontal shipper collaborations that translates to 
(constrained variants of) the lane covering problem (LCP).  
Formally, LCP is a covering problem on a directed graph that
determines a set of simple cycles covering all lanes to mimimize the total travel cost, 
where lane denotes the arc between the pickup and delivery nodes 
of a full truckload (FTL) request from a shipper. 
While LCP can be solved in polynomial time 
(because it can be formulated as a minimum cost circulation problem),  
it becomes NP-hard 
when more realistic constraints on the cycles are cinsidered (Ergun et al.~\cite{ozlem-b}). 
Therefore, these studies have focussed on 
efficient approximate-solution methods for additional constraints such as  
cardinalities (Ergun et al.~\cite{ozlem-b}), 
time windows (Ergun et al.~\cite{ozlem-a}, Ghiani et al.~\cite{Ghiani2008}), and 
partner bounds (Kuyzu~\cite{Kuyzu2017}). 
However, even if each company specifies detailed matching conditions in advance and the system finds and presents a partner combination that mutually satisfies these conditions, 
the proposal is often not accepted because it violates some potential constraints 
known only to each company. 
Therefore, the authors are interested in enumerating coalitions (i.e., combinations of lanes) rather than in finding a partition of lanes (i.e., coalition structure generation). 
To the best of our knowledge, The closest study is that by Creemers et al.~\cite{Creemers2017}, who proposed the combined use of one-dimensional sort\footnote{lanes are sorted on the latitudinal coordinate only.} and a bounding-box approach to quickly identifying geographically close lanes with the intent of 
detecting collaborative shipping opportunities.
In contrust, 
we use a completely different approach and enumerate instantaneously 
the combinations of triangular transports whose occupied vehicle rate exceeds a given value.
%
\section{Triangular Transport Enumeration Problem} 
\label{Sec:formulation}
%

\subsection{Problem formulation} 
We are given a metric space $(B, d)$,  
where $B$ is a finite set of transportation bases and 
$d : B \!\times\! B \to \mathbb{R}_{\geq 0}$ is a distance function. 
For simplicity, a full truckload request from one base to another is called a transport lane (or lane),   
and a finite set of lanes $T$ is provided. 
For every lane $t \in T$, we denote the start point and the endpoint as $t^{\rm s}$ and $t^{\rm e}$, respectively.
Furthermore, for every lane $t \in T$, we denote the distance of $t$ by $t^{\rm d}$, where $t^{\rm d} = d(t^{\rm s}, t^{\rm e})$.
\par
A series of transportation in which a single truck sequentially handles three different lanes 
$t_1, t_2$, and $t_3$ in this order is called a \lq\lq triangular transport" and denoted by $(t_1, t_2, t_3)$.  
Given a triangular transport $(t_1, t_2, t_3)$, we use symbols $d_i$ and $e_i$, with $i = 1,2,3$, in the following manner:
\begin{align*}
d_i &= t_i^{\rm d}, \quad i = 1,2,3,  \\
e_1 &= d(t_{1}^{\rm e}, t_{2}^{\rm s}), \quad e_2 = d(t_{2}^{\rm e}, t_{3}^{\rm s}), \quad
e_3 = d(t_3^{\rm e}, t_{1}^{\rm s}).
\end{align*}
That is, $d_i$ represents the distance of the $i$th loading trip and $e_i$ represents the 
distance of $i$th empty trip (see Figure~\ref{fig:d_and_e}). 
\begin{figure}[H]
\centering
\includegraphics[scale=0.11]{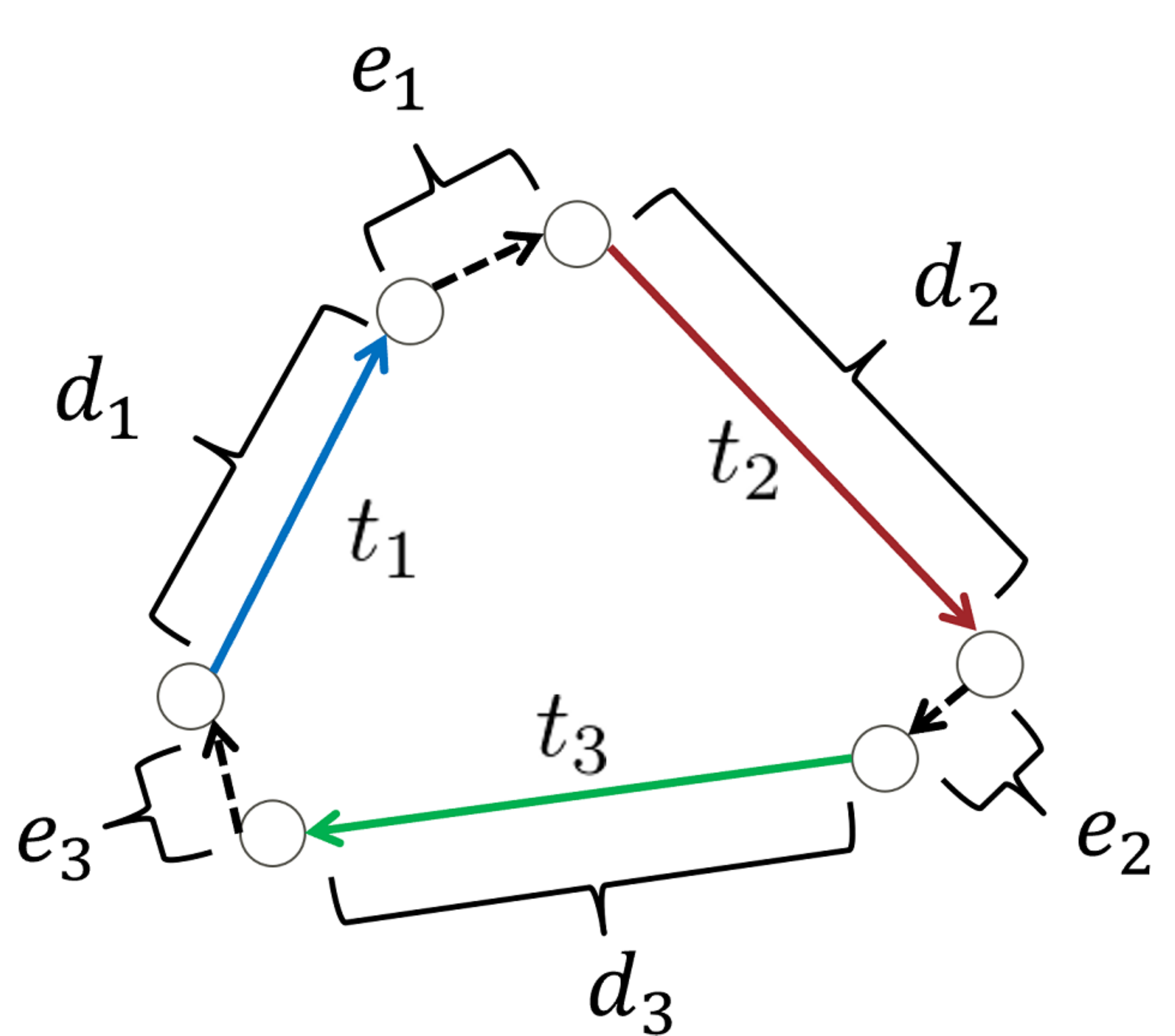} 
\caption{Symbols $d_i$ and $e_i$ representing distances}
\label{fig:d_and_e}
\end{figure}

Using the above notation, 
the occupied vehicle rate and the total mileage for the triangular transport can be expressed as follows:
\begin{align*}
& \textrm{Occupied vehicle rate} = \frac{d_1 + d_2 + d_3}{d_1 + e_1 + d_2 + e_2 + d_3 + e_3}, \\[2mm]
& \textrm{Total mileage} = d_1 + e_1 + d_2 + e_2 + d_3 + e_3. 
\end{align*}
For a given transport lane $t_1 \in T$, we search for joint transport partners $t_2, t_3 \in T$ 
and propose an effective triangular transport $(t_1, t_2, t_3)$ (see Figure 3). 
In particular, our goal is to list all the triangular transports that satisfy the following two conditions:
\begin{df}
For any  any $\ell \in (0, 1]$ and for any $u > 0$, a triangular transport is said to be $(\ell, u)$-feasible 
if it satisfies the following two constraints.
\begin{itemize}
\item[\rm (C1)] 
The occupied vehicle rate is $\ell$ or higher, Namely,
\begin{equation*}
\frac{d_1 + d_2 + d_3}{d_1 + e_1 + d_2 + e_2 + d_3 + e_3} \geqq \ell.
\end{equation*}
\item[\rm (C2)] 
The total mileage is $u$ or less. Namely,
\begin{equation*}
d_1 + e_1 + d_2 + e_2 + d_3 + e_3 \leqq u.
\end{equation*}
\end{itemize}
\end{df}

\begin{figure}[H]
\centering
\includegraphics[scale=0.15]{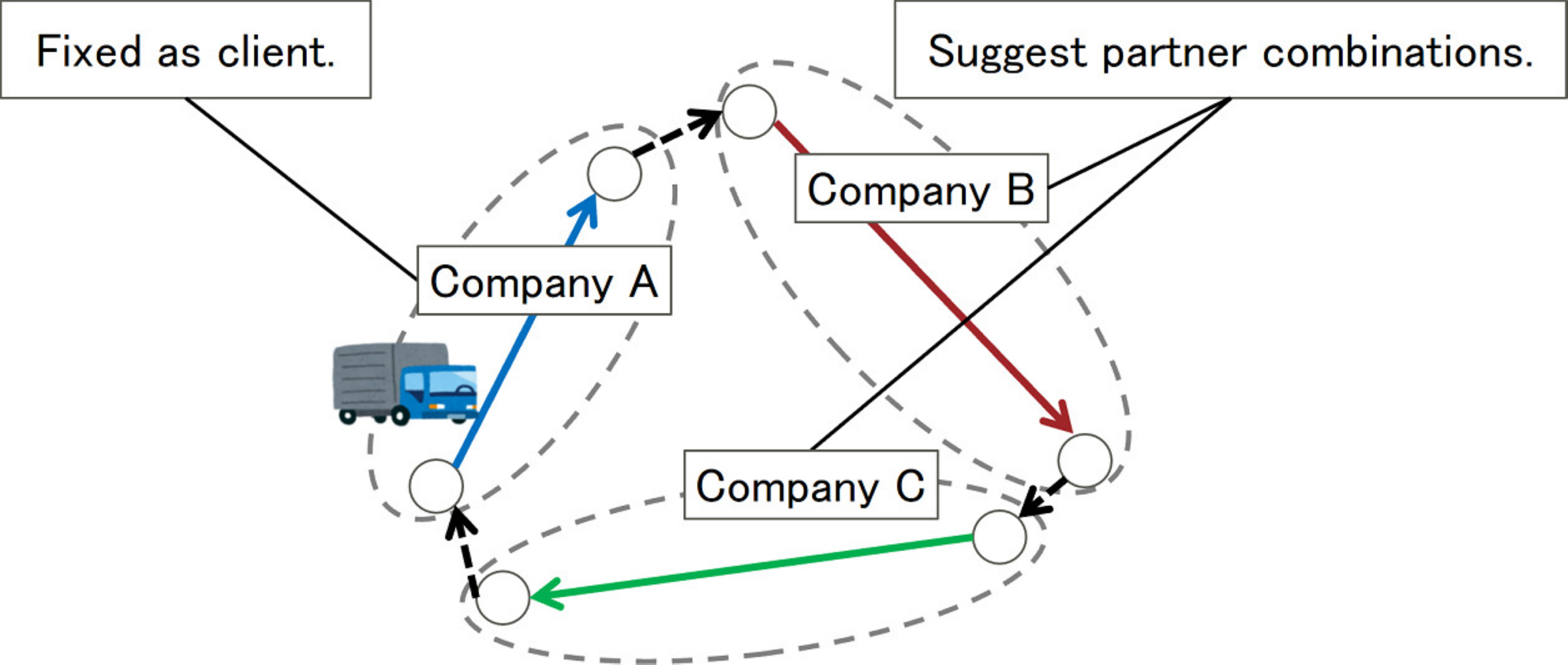} 
\caption{Processing a triangular transport matching request}
\label{fig:system}
\end{figure}

\begin{rem}
The three triangular transports $(t_1, t_2, t_3), (t_2, t_3, t_1)$, and $(t_3, t_1, t_2)$ 
obtained by cyclic permutation are essentially equivalent in the sense that they must result in the same occupied vehicle rate and the same total mileage.
Therefore, we can fix $t_1$ as a matching client without loss of generality. 
However, the cost of triangular transportation varies depending on the point of origin, and this difference is taken into account separately when estimating the costs. 
\end{rem}

In general, the higher the occupied vehicle rate is, the more efficient the process is. 
However, when $t_1$ is a short-haul lane from Tokyo to Saitama, selecting a long-haul lane from Saitama to Hokkaido as $t_2$ and a long-haul lane from Hokkaido to Tokyo as $t_3$ is not recommended, no matter how high the occupied vehicle rate is. In other words, 
only $t_2$ and $t_3$ should cooperate for round-trip transportation. The total mileage limit also serves to exclude such undesired combinations. 
\par
As the frequency of transportation and its seasonal variations must also be taken into consideration when selecting partners, it is convenient for the client to be presented with multiple candidates.
This is why we focus on enumeration rather than maximizing the occupied vehicle rate.
However, if the system presents too many candidates, 
the matching client cannot consider them deeply. Therefore, in practice, it is sufficient to enumerate and present $k$ combinations (for a suitable $k$) from those with the highest occupied vehicle rates. 
In this study, we also consider such the top-$k$ searches.

\subsection{Simple brute-force method and its drawbacks} 
For a given transport lane $t_1$, we only need to search for partners $t_2$ and $t_3$. 
Therefore we can consider a simple brute-force search with a double for-loop, 
as shown in Algorithm~\ref{Alg:simple}. 
However, if 10,000 lanes are registered in the database, 
 there are approximately 100 million combinations of 
 two lanes. As it is also necessary to calculate the distance traveled, 
 which differs according to the combination, a simple brute-force calculation using a calculator 
 would take a long time.  In reality, logistics companies have multiple transport lanes 
 (a single company may have hundreds or even thousands of lanes). 
 Thus, it is difficult for conventional methods to respond instantly to multiple matching requests 
 from a single user or a series of requests from multiple users.
\par
Hence, the remaining task discussed in this study is to provide a practically faster algorithm
for enumerating $(\ell, u)$-feasible triangular transports.

\begin{algorithm2e}[h]
\SetKwData{Input}{input}
\KwData{a set of lanes $T$ on a metric space $(B, d)$, a lane $t_1$, a desired occupied vehicle rate $\ell \in [0, 1]$, and an upper limit of mileage $u > 0$}
\KwResult{the set of all $(\ell, u)$-feasible triangular transports containing $t_1$}

$C \leftarrow \emptyset$\;
$d_1 \leftarrow t_1^{\rm d}$\;
\ForAll{$s \in S$}{
    $e_1 \leftarrow d(t_1^{\rm e}, s)$\;
    \ForAll{$t_2 \in T(s) \setminus \{t_1\}$}{
        $d_2 \leftarrow t_2^{\rm d}$\;
        \ForAll{$s' \in S$}{
            $e_2 \leftarrow d(t_2^{\rm e}, s')$\;
            \ForAll{$t_3 \in T(s') \setminus \{t_1, t_2\}$}{
                $d_3 \leftarrow t_3^{\rm d}$\;
                $e_3 \leftarrow d(t_3^{\rm e}, t_1^{\rm s})$\;
                \If{$d_1 + e_1 + d_2 + e_2 + d_3 + e_3 \leqq u$ and $\ell \leqq \frac{d_1 + d_2 + d_3}{d_1 + e_1 + d_2 + e_2 + d_3 + e_3}$}{
                    $C \leftarrow C \cup \{(t_1, t_2, t_3)\}$\;
                }
            }
        }
    }
}
\caption{A simple brute-force search}
\label{Alg:simple}
\end{algorithm2e}
%
\section{Pruning Algorithm} 
\label{Sec:pruning}
%

\subsection{Quadruple looping of brute-force search} 
For any transportation base $b \in B$, let $T(b)$ denote the set of all lanes 
starting from $b$. In addition, let $S$ be the set of all bases that can be the starting point for a lane. 
That is,
\begin{equation*}
S := \{ b \in B \,|\, T(b) \neq \emptyset \}.
\end{equation*}
If we partition $T$ into $T = \cup_{s \in S} T(s)$ in advance, then 
the for loop traversing $T$ can be replaced by a double for loop 
(outer for loop traversing $S$ and inner for loop traversing $T(s)$).
Thus, Algorithm~\ref{Alg:simple}, which is written as a double for loop, 
can be rewritten as a quadruple for loop, as shown in Algorithm~\ref{Alg:quadruple}.
We consider effectively narrowing down the search range for each of the four for loops 
without sacrificing accuracy.

\begin{algorithm2e}[H]
\SetKwData{Input}{input}
\KwData{a set of lanes $T$ on a metric space $(B, d)$, a lane $t_1$, a desired occupied vehicle rate $\ell \in [0, 1]$, and an upper limit of mileage $u > 0$}
\KwResult{the set of all $(\ell, u)$-feasible triangular transports containing $t_1$}

$C \leftarrow \emptyset$\;
$d_1 \leftarrow t_1^{\rm d}$\;
\ForAll{$s \in S$}{
    $e_1 \leftarrow d(t_1^{\rm e}, s)$\;
    \ForAll{$t_2 \in T(s) \setminus \{t_1\}$}{
        $d_2 \leftarrow t_2^{\rm d}$\;
        \ForAll{$s' \in S$}{
            $e_2 \leftarrow d(t_2^{\rm e}, s')$\;
            \ForAll{$t_3 \in T(s') \setminus \{t_1, t_2\}$}{
                $d_3 \leftarrow t_3^{\rm d}$\;
                $e_3 \leftarrow d(t_3^{\rm e}, t_1^{\rm s})$\;
                \If{$d_1 + e_1 + d_2 + e_2 + d_3 + e_3 \leqq u$ and $\ell \leqq \frac{d_1 + d_2 + d_3}{d_1 + e_1 + d_2 + e_2 + d_3 + e_3}$}{
                    $C \leftarrow C \cup \{(t_1, t_2, t_3)\}$\;
                }
            }
        }
    }
}
\caption{Quadruple looping of brute-force search}
\label{Alg:quadruple}
\end{algorithm2e}

\subsection{Dynamic refinement of search range} 
When pruning the first for loop in Algorithm~\ref{Alg:quadruple}, while we can determine the value of $d_1$, 
we cannot determine the values of 
$d_2, e_2, d_3, e_3$,  considering these values can be changed using the inner for loops.
In this situation, we introduce the following lemma for efficient pruning: 

\begin{lem}
A necessary condition for a triangular transport to be $(\ell, u)$-feasible is
\begin{equation*}
e_1 \leqq \min \{u(1 - \ell), u - d_1\}.
\end{equation*}
\label{Lem:pruning1}
\end{lem}

\begin{proof}
We can easily find a lower bound of the total mileage as follows:
\begin{equation*}
d_1 + e_1 \leqq d_1 + e_1 + d_2 + e_2 + d_3 + e_3 \leqq u.
\end{equation*}
Therefore, we have
\begin{equation}
e_1 \leqq u - d_1.
\label{1a}
\end{equation}
On the other hand, we can find an upper bound of the occupied vehicle rate as follows:
\begin{equation*}
\ell \leqq \frac{d_1 + d_2 + d_3}{d_1 + e_1 + d_2 + e_2 + d_3 + e_3} \leqq 1 - \frac{e_1}{u}.
\end{equation*}
Thus, by solving $\ell \leqq 1 - \frac{e_1}{u}$ with respect to $e_1$, we obtain
\begin{equation}
e_1 \leqq u(1 - \ell).
\label{1b}
\end{equation}
The result follows from (\ref{1a}) and (\ref{1b}).
\end{proof}

Similarly, when pruning the second for loop in Algorithm~\ref{Alg:quadruple}, 
we can determine the values of $d_1, e_1$, but not  
$e_2, d_3, e_3$.  Therefore, we introduce the following lemma for efficient pruning.

\begin{lem}
A necessary condition for a triangular transport to be $(\ell, u)$-feasible is
\begin{equation*}
d_2 \geqq \left\{
\begin{array}{cl}
\frac{2\ell - 1}{2(1 - \ell)}e_1 - d_1 & (\ell \neq 1) \\
0 & (\ell = 1),
\end{array}
\right.
\quad
d_2 \leqq u - (d_1 + e_1).
\end{equation*}
\label{Lem:pruning2}
\end{lem}

\begin{proof}
We can easily find a lower bound of the total mileage as follows:
\begin{equation*}
d_1 + e_1 + d_2 \leqq d_1 + e_1 + d_2 + e_2 + d_3 + e_3 \leqq u.
\end{equation*}
Therefore, we have
\begin{equation}
d_2 \leqq u - (d_1 + e_1).
\label{2a}
\end{equation}
On the other hand, we can find an upper bound of the occupied vehicle rate as follows:
\begin{align*}
\ell &\leqq \frac{d_1 + d_2 + d_3}{d_1 + e_1 + d_2 + e_2 + d_3 + e_3}  && \textrm{(monotonically decreasing w.r.t. $e_2, e_3$)} \\
&\leqq \frac{d_1 + d_2 + d_3}{d_1 + e_1 + d_2 + d_3} && \\
&= 1 - \frac{e_1}{d_1 + e_1 + d_2 + d_3} &&\textrm{(monotonically increasing w.r.t. $d_3$)}  \\
&\leqq 1 - \frac{e_1}{d_1 + e_1 + d_2 + (d_1 + e_1 + d_2)} && \\
&= 1 - \frac{e_1}{2(d_1 + e_1 + d_2)}. &&
\end{align*}
In the above, to obtain the second inequality, we replace $e_2$ and $e_3$ with zero, which is a lower bound of 
$e_2$ and $e_3$. Similarly, to obtain the third inequality, we replace $d_3$ with $(d_1 + e_1 + d_2)$, 
which is an upper bound of $d_3$ under the constraints $e_2 = e_3 = 0$. 
This upper bound is derived from the metric axioms (Figure~\ref{fig:inequality}). 
By solving the final inequality $\ell \leqq 1 - \frac{e_1}{2(d_1 + e_1 + d_2)}$ with respect to $d_2$ when $\ell \neq 1$,  we obtain
\begin{equation}
\frac{2\ell - 1}{2(1 - \ell)}e_1 - d_1 \leqq d_2 \quad (\ell \neq 1).
\label{2b}
\end{equation}
The result follows from (\ref{2a}) and (\ref{2b}).
\end{proof}

\begin{rem}
In Lemma~\ref{Lem:pruning2}, we use not only the triangle inequality but also the symmetry of the distance to obtain the upper bound of $d_3$.
\end{rem}

\begin{figure}[H]
\centering
\includegraphics[scale=0.11]{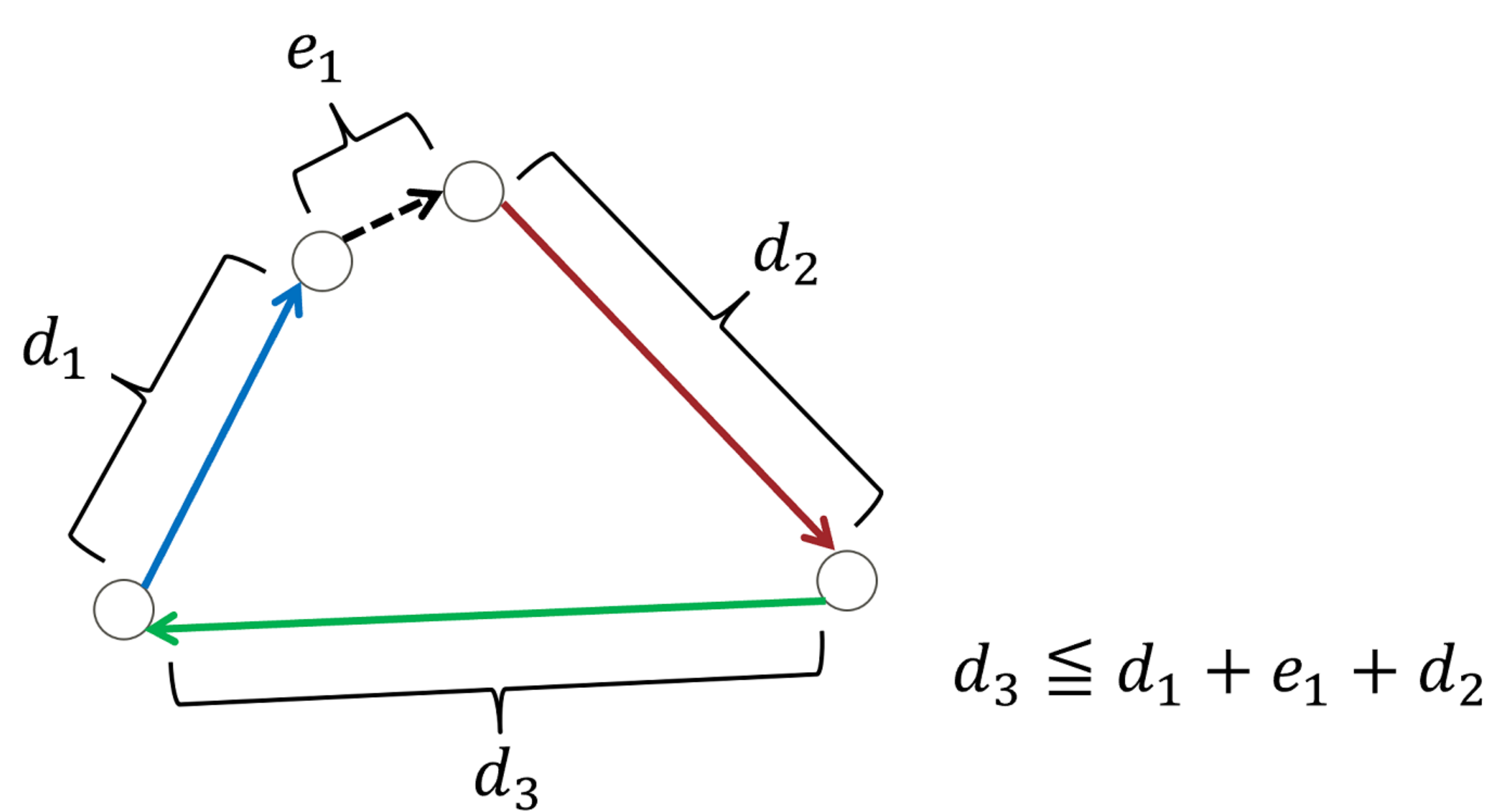} 
\caption{An upper bound of $d_3$ derived from the distance axiom under $e_2 = e_3 = 0$.}
\label{fig:inequality}
\end{figure}

When pruning the third for loop in Algorithm~\ref{Alg:quadruple}, 
we can determine the values of $d_1, e_1, d_2$, but not  
$d_3, e_3$. Thus, we introduce the following lemma for efficient pruning.

\begin{lem}
A necessary condition for a triangular transport to be $(\ell, u)$-feasible is
\begin{equation*}
e_2 \leqq \min \{u(1 - \ell) - e_1, u - (d_1+ e_1 + d_2)\}.
\end{equation*}
\label{Lem:pruning3}
\end{lem}

\begin{proof}
We can easily find a lower bound of the total mileage as follows:
\begin{equation*}
d_1 + e_1 + d_2 + e_2 \leqq d_1 + e_1 + d_2 + e_2 + d_3 + e_3 \leqq u.
\end{equation*}
Therefore, we have
\begin{equation}
e_2 \leqq u - (d_1 + e_1 + d_2).
\label{3a}
\end{equation}
On the other hand, we can find an upper bound of the occupied vehicle rate as follows:
\begin{equation*}
\ell \leqq \frac{d_1 + d_2 + d_3}{d_1 + e_1 + d_2 + e_2 + d_3 + e_3} \leqq 1 - \frac{e_1 + e_2}{u}.
\end{equation*}
Thus, by solving $\ell \leqq 1 - \frac{e_1 + e_2}{u}$ with respect to $e_2$, we obtain
\begin{equation}
e_2 \leqq u(1 - \ell) - e_1.
\label{3b}
\end{equation}
The result follows from (\ref{3a}) and (\ref{3b}).
\end{proof}

Finally, when pruning the fourth for loop in Algorithm~\ref{Alg:quadruple}, 
we can determine the values of $d_1, e_1, d_2, e_2$, but not the value of 
$e_3$.  Therefore, we introduce the following lemma for efficient pruning.

\begin{lem}
A necessary condition for a triangular transport to be $(\ell, u)$-feasible is
\begin{equation*}
d_3 \geqq 
\left\{
\begin{array}{cl}
\frac{\ell}{1 - \ell}(e_1 + e_2) - (d_1 + d_2) & (\ell \neq 1) \\
0 & (\ell = 1), 
\end{array}
\right.
\quad d_3 \leqq u - (d_1 + e_1 + d_2 + e_2).
\end{equation*}
\label{Lem:pruning4}
\end{lem}

\begin{proof}
We can easily find a lower bound of the total mileage as follows:
\begin{equation*}
d_1 + e_1 + d_2 + e_2 + d_3 \leqq d_1 + e_1 + d_2 + e_2 + d_3 + e_3 \leqq u.
\end{equation*}
Therefore, we have
\begin{equation}
d_3 \leqq u - (d_1 + e_1 + d_2 + e_2).
\label{4a}
\end{equation}
On the other hand, we can find an upper bound of the occupied vehicle rate as follows:
\begin{align*}
\ell &\leqq \frac{d_1 + d_2 + d_3}{d_1 + e_1 + d_2 + e_2 + d_3 + e_3}  \qquad \textrm{(monotonically decreasing w.r.t. $e_3$)} \\
&\leqq \frac{d_1 + d_2 + d_3}{d_1 + e_1 + d_2 + e_2 + d_3} \\
&= 1 - \frac{e_1 + e_2}{d_1 + e_1 + d_2 + e_2 + d_3}.
\end{align*}
In the above, to get the second inequality, we replace $e_3$ with zero, which is a lower bound of 
$e_3$. 
By solving the final inequality $\ell \leqq 1 - \frac{e_1 + e_2}{d_1 + e_1 + d_2 + e_2 + d_3}$ with respect to $d_3$, we obtain
\begin{equation}
\frac{\ell}{1 - \ell}(e_1 + e_2) - (d_1 + d_2) \leqq d_3 \quad (\ell \neq 1).
\label{4b}
\end{equation}
The result follows from (\ref{4a}) and (\ref{4b}).
\end{proof}

The proposed algorithm that incorporates the above discussion is presented in Algorithm~\ref{Alg:pruning}.
In order to efficiently scan the first and third for-loops in Algorithm~\ref{Alg:pruning}, 
for every $b \in B$, we can sort all transportation bases in $S$ according to their distances from $b$ and create a sorted list $S_b$ in advance.
Since the sort is repeated $|B|$ times, this operation takes $O(|B|^2 \log |B|)$ time. 
However, owing to the nature of the problem, the number of transportation bases $|B|$ is small compared with the number of transports $|T|$. 
Therefore, this preprocessing can be performed in a relatively short time. 
In addition, in order to efficiently scan the second and fourth for loops, for every $s \in S$, 
we sort all the lanes in $T(s)$ according to their distances of that lane in advance.

\begin{algorithm2e}[p]
\SetKwData{Input}{input}
\KwData{a set of lanes $T$ on a metric space $(B, d)$, a lane $t_1$, a desired occupied vehicle rate $\ell \in [0, 1]$, and an upper limit of mileage $u > 0$}
\KwResult{the set of all $(\ell, u)$-feasible triangular transports containing $t_1$}

$C \leftarrow \emptyset$\;
$d_1 \leftarrow t_1^{\rm d}$\;
\ForAll{$s \in S$ such that $d(t_1^{\rm e}, s) \leqq \min \{u(1 - \ell), u - d_1\}$}{
    $e_1 \leftarrow d(t_1^{\rm e}, s)$\;
    \If{$u < d_1 + e_1 + d(s, t_1^{\rm s})$}{
        continue\;
    }
    \ForAll{$t_2 \in T(s) \setminus \{t_1\}$ such that $t_2^{\rm d} \geqq \frac{2\ell - 1}{2(1 - \ell)}e_1 - d_1 \; (\textrm{if $\ell \neq 1$}), \;\, t_2^{\rm d} \leqq u - (d_1 + e_1)$}{
        $d_2 \leftarrow t_2^{\rm d}$\;
        \If{$u < d_1 + e_1 + d_2 + d(t_2^{\rm e}, t_1^{\rm s})$}{
            continue\;
        }
        \ForAll{$s' \in S$ such that $d(t_2^{\rm e}, s') \leqq \min \{u(1 - \ell) - e_1, u - (d_1 + e_1 + d_2)\}$}{
            $e_2 \leftarrow d(t_2^{\rm e}, s')$\;
            \If{$u < d_1 + e_1 + d_2 + e_2 + d(s', t_1^{\rm s})$}{
                continue\;
            }
            \ForAll{$t_3 \in T(s') \setminus \{t_1, t_2\}$ such that $t_3^{\rm d} \geqq \frac{\ell}{1 - \ell}(e_1 + e_2) - (d_1 + d_2) \; (\textrm{if $\ell \neq 1$}), \;\, t_3^{\rm d} \leqq u - (d_1 + e_1 + d_2 + e_2)$}{
                $d_3 \leftarrow t_3^{\rm d}$\;
                $e_3 \leftarrow d(t_3^{\rm e}, t_1^{\rm s})$\;
                \If{$d_1 + e_1 + d_2 + e_2 + d_3 + e_3 \leqq u$ and $\ell \leqq \frac{d_1 + d_2 + d_3}{d_1 + e_1 + d_2 + e_2 + d_3 + e_3}$}{
                    $C \leftarrow C \cup \{(t_1, t_2, t_3)\}$\;
                }
            }
        }
    }
}
\caption{A brute-force search with pruning}
\label{Alg:pruning}
\end{algorithm2e}

\begin{thm}
Algorithm~\ref{Alg:pruning} correctly outputs the set of all $(\ell, u)$-feasible triangular transports containing $t_1$. 
\end{thm}

\begin{proof}
The result follows from Lemmas~\ref{Lem:pruning1} to \ref{Lem:pruning4}.
\end{proof}

In Algorithm~\ref{Alg:pruning}, when $\ell = 1$, 
enumeration is the easiest because the subsequent departure locations searched in the first and third for-loops are limited to those with zero empty trip distances. 
In addition, when $\ell \neq 1$, the closer the value of $\ell$ is to $1$, 
the smaller is the search interval limited by \lq\lq such that" clause for each of the four for loops. 
Hence, we have the following remark:

\begin{rem}
The larger the desired occupied vehicle rate $\ell$ is, 
the shorter the Algorithm~\ref{Alg:pruning} runs.
\label{Rem:shorter}
\end{rem}
\newpage
\subsection{Faster algorithm for the $k$-best solutions} 
In practice, it is sufficient to present a list of $k$ options with the highest occupied vehicle rate for a suitable $k$, considering the matching client cannot consider them deeply if the system presents too many options. 
More precisely, we enumerate the $k$ triangular transports in the order of the highest occupied vehicle rate among the $(\ell, u)$-feasible triangular transports containing $t_1$. 
Therefore, when searching for the $(\ell, u)$-feasible triangular transports containing $t_1$, 
we always keep track of the occupied vehicle rate, which is the provisional $k$th rank. 
Until the number of $(\ell, u)$-feasible triangular transports reaches $k$, 
the desired occupied vehicle rate $\ell$ is used as it is. 
Thereafter, each time the provisional $k$th position is updated, 
the value of $\ell$ used in the algorithm is increased to the value of the occupied vehicle rate 
in the provisional $k$th position. 
Then, triangular transport with an occupied vehicle rate up to the top $k$th will always be covered, and the modified algorithm will work in less time than continuing to use the original value of $\ell$. 
This management can be performed using a priority queue (Here we use a binary heap~\cite{heap}). 
The detailed procedure is presented as Algorithm~\ref{Alg:dynamic}.

\begin{algorithm2e}[h]
\SetKwData{Input}{input}
\KwData{a set of lanes $T$ on a metric space $(B, d)$, a lane $t_1$, a desired occupied vehicle rate $\ell \in [0, 1]$, and an upper limit of mileage $u > 0$}
\KwResult{the set of all $(\ell, u)$-feasible triangular transports containing $t_1$}
Declare a binary heap (min-heap) $H$\;
$d_1 \leftarrow t_1^{\rm d}$\;
\ForAll{$s \in S$ such that $d(t_1^{\rm e}, s) \leqq \min \{u(1 - \ell), u - d_1\}$}{
    $e_1 \leftarrow d(t_1^{\rm e}, s)$\;
    \If{$u < d_1 + e_1 + d(s, t_1^{\rm s})$}{
        continue\;
    }
    \ForAll{$t_2 \in T(s) \setminus \{t_1\}$ such that $t_2^{\rm d} \geqq \frac{2\ell - 1}{2(1 - \ell)}e_1 - d_1 \; (\textrm{if $\ell \neq 1$}), \;\, t_2^{\rm d} \leqq u - (d_1 + e_1)$}{
        $d_2 \leftarrow t_2^{\rm d}$\;
        \If{$u < d_1 + e_1 + d_2 + d(t_2^{\rm e}, t_1^{\rm s})$}{
            continue\;
        }
        \ForAll{$s' \in S$ such that $d(t_2^{\rm e}, s') \leqq \min \{u(1 - \ell) - e_1, u - (d_1 + e_1 + d_2)\}$}{
            $e_2 \leftarrow d(t_2^{\rm e}, s')$\;
            \If{$u < d_1 + e_1 + d_2 + e_2 + d(s', t_1^{\rm s})$}{
                continue\;
            }
            \ForAll{$t_3 \in T(s') \setminus \{t_1, t_2\}$ such that $t_3^{\rm d} \geqq \frac{\ell}{1 - \ell}(e_1 + e_2) - (d_1 + d_2) \; (\textrm{if $\ell \neq 1$}), \;\, t_3^{\rm d} \leqq u - (d_1 + e_1 + d_2 + e_2)$}{
                $d_3 \leftarrow t_3^{\rm d}$\;
                $e_3 \leftarrow d(t_3^{\rm e}, t_1^{\rm s})$\;
                \If{$d_1 + e_1 + d_2 + e_2 + d_3 + e_3 \leqq u$ and $\ell \leqq \frac{d_1 + d_2 + d_3}{d_1 + e_1 + d_2 + e_2 + d_3 + e_3}$}{ 
                    \If{$|H| = k$}{
                        Remove the minimum element from $H$\;
                    }
                    \vspace{0mm} Add $(t_1, t_2, t_3)$ to $H$ with priority $\frac{d_1 + d_2 + d_3}{d_1 + e_1 + d_2 + e_2 + d_3 + e_3}$ \;
                    \If{$|H| = k$}{
                        $\ell \leftarrow$ Refer to the minimum key of $H$ without deleting it\;
                    }
                }
            }
        }
    }
}
\caption{A brute-force top-$k$ search with pruning}
\label{Alg:dynamic}
\end{algorithm2e}
\newpage
%
\section{Computational Experiments}
%
Using approximately 17,000 real, anonymized transport lane data ($|B| = 4828, \;\, |T| = 16957$) 
across Japan, 
we conducted experiments to enumerate efficient triangular transports. 
First we created 1000 problems (1000 matching requests) 
by randomly selecting 1000 lanes from $T$ and fixing each lane as the first lane. 
These problems were used to compare three algorithms. 
Algorithm~\ref{Alg:simple} performs a simple brute-force search,  
Algorithm~\ref{Alg:pruning} dynamically narrows the search range,  
and Algorithm~\ref{Alg:dynamic} specializes in enumerating the $k$ best solutions to Algorithm~\ref{Alg:pruning}. 
The desired occupied vehicle rate $\ell$ was varied from 0.75 to 0.95 in 0.05 increments. 
For the total mileage limit, we set $u = 4 t_1^{\rm d}$, 
that is, the total mileage can be up to four times the length of the first transport lane, which is fixed as a matching client. 
We implemented the three algorithms using Cython~\cite{cython, cython-guide} ($\neq$ Python) and executed them on a desktop PC
with an Intel\textregistered~Core\texttrademark~i9-9900K processor and 64GB memory installed. 
The results are presented in Tables~\ref{Tbl:result1} and \ref{Tbl:result2}.

\begin{table}[h]
\centering
\caption{Computational time to process 1000 matching requests (seconds)} 
\label{Tbl:result1}
\begin{tabular}{|c|r|r|r|} \hline
\multirow{2}{*}{Scenario \textbackslash Algorithm}  & Brute-force search & pruning & $k$ Best solutions \\
& Algorithm~\ref{Alg:simple}  & Algorithm~\ref{Alg:pruning}  & Algorithm~\ref{Alg:dynamic}  \\ \hline
$\ell = 0.75, \;\, u = 4t_1^{\rm d}$ & 103635.7 & 1525.7 & 29.3  \\ 
$\ell = 0.80, \;\, u = 4t_1^{\rm d}$ & (28h 47min) & 944.2 & 28.7  \\ 
$\ell = 0.85, \;\, u = 4t_1^{\rm d}$ & $\downarrow$  \quad & 512.3 & 27.0  \\ 
$\ell = 0.90, \;\, u = 4t_1^{\rm d}$ & $\downarrow$ \quad & 209.3 & 22.1  \\ 
$\ell = 0.95, \;\, u = 4t_1^{\rm d}$ & $\downarrow$ \quad & 45.0 & 13.1  \\ \hline
\end{tabular}
\end{table}

\begin{table}[h]
\centering
\caption{Average computational time to process one matching request (seconds)} 
\label{Tbl:result2}
\begin{tabular}{|c|r|r|r|} \hline
\multirow{2}{*}{Scenario \textbackslash Algorithm}  & Brute-force search & pruning & $k$ Best solutions \\
& Algorithm~\ref{Alg:simple}  & Algorithm~\ref{Alg:pruning}  & Algorithm~\ref{Alg:dynamic}  \\ \hline
$\ell = 0.75, \;\, u = 4t_1^{\rm d}$ & 103.6357 & 1.5257 & 0.0293  \\ 
$\ell = 0.80, \;\, u = 4t_1^{\rm d}$ & $\downarrow$ \quad & 0.9442 & 0.0287  \\ 
$\ell = 0.85, \;\, u = 4t_1^{\rm d}$ & $\downarrow$ \quad & 0.5123 & 0.0270  \\ 
$\ell = 0.90, \;\, u = 4t_1^{\rm d}$ & $\downarrow$ \quad & 0.2093 & 0.0221  \\ 
$\ell = 0.95, \;\, u = 4t_1^{\rm d}$ & $\downarrow$ \quad & 0.0450 & 0.0131  \\ \hline
\end{tabular}
\end{table}

Using the simple brute force, 
it took approximately 100 seconds to process a matching request.  
It is also clear that 1,000 matching requests cannot be processed in a single day. 
On the other hand, Algorithm~\ref{Alg:pruning} 
processed one matching request in a few seconds on average, 
and processed 1000 matching requests in less than 30 minutes 
(even in the most time-consuming case, $\ell = 0.75$). 
As mentioned in Remark~\ref{Rem:shorter}, we observed that as the value of the desired occupied vehicle rate $\ell$ increased, 
Algorithm~\ref{Alg:pruning} ran in a shorter time.
Algorithm~\ref{Alg:dynamic} 
instantly processed one matching request on average, and 1000 matching requests were processed within 30 seconds. 
Furthermore the computational time does not change significantly even when the desired occupied vehicle rate is low.
\par
We believe that Algorithm~\ref{Alg:pruning} is sufficiently practical to solve a range of valid scenarios. 
In reality, triangular transportation with an occupied vehicle rate of less than 80\% is unattractive.  
It is not meaningful to enumerate the range of such transports. 
However, the problem is setting the desired occupied vehicle rate when using 
Algorithm~\ref{Alg:pruning} in actual service. 
Although it is natural for users to set a value according to their tolerance, 
if there are insufficient triangular transports to achieve a set value, 
the user must change the setting to a lower value and start over. 
Therefore, an incentive exists to set a low value from the beginning. 
Consequently, the system is computationally overloaded. 
Conversely, even if the system sets the value rather than the user, it is not easy to determine the tradeoff between the number of triangular transports to be enumerated and the computational time in advance. 
After all, the form of listing the top-$k$ combinations from among those with the highest occupied vehicle rates is easy to understand for both the user and the system, 
which can be achieved using Algorithm~\ref{Alg:dynamic}. Even if the desired occupied vehicle rate is set to a low value, the computational load does not change significantly.
Therefore, Algorithm~\ref{Alg:dynamic} is also a means of solving the operational problems in Algorithm~\ref{Alg:pruning}.
\par
In the computational experiments, we recorded how far the value of $\ell$ used inside Algorithm~\ref{Alg:dynamic} was increased. 
Hereafter, we write $\ell^*$ as the value of $\ell$ at the completion of the algorithm operation. 
When the original desired occupied vehicle rate was set to $\ell = 0.75$,
 the average value of $\ell^*$ for the 1000 problems was 0.875. 
Nevertheless, the computational time required for Algorithm~\ref{Alg:dynamic} with $\ell = 0.75$ to process 1000 matching requests (29.3 seconds) was shorter than that of Algorithm~\ref{Alg:pruning} with $\ell = 0.95$ (45.0 seconds). 
Figure~\ref{fig:boxplot} shows a box-and-whisker plot of the distribution of the computation time for both. 
The Algorithm~\ref{Alg:pruning} with $\ell = 0.95$ could instantly process most of the 1000 matching requests, however, a small fraction of the matching requests worsened the overall computational time. 
We examined the characteristics of these matching requests and found that the first transport lane, 
which was fixed as the client, was long distance, and there were many other transportation bases near the end of the first lane, 
making it difficult to ensure effective pruning. 
To match requests with these characteristics, there are many triangular transport combinations with occupied vehicles rate close to one. 
Algorithm~\ref{Alg:dynamic} with $\ell = 0.75$ reduces the deterioration in computational time for such matching requests by increasing the value of $\ell$ used in the algorithm to a higher value. 
In other words, Algorithm~\ref{Alg:dynamic} is superior in terms of its strategy for instances that are inherently difficult to enumerate.

\begin{figure}[H]
\centering
\includegraphics[scale=0.16]{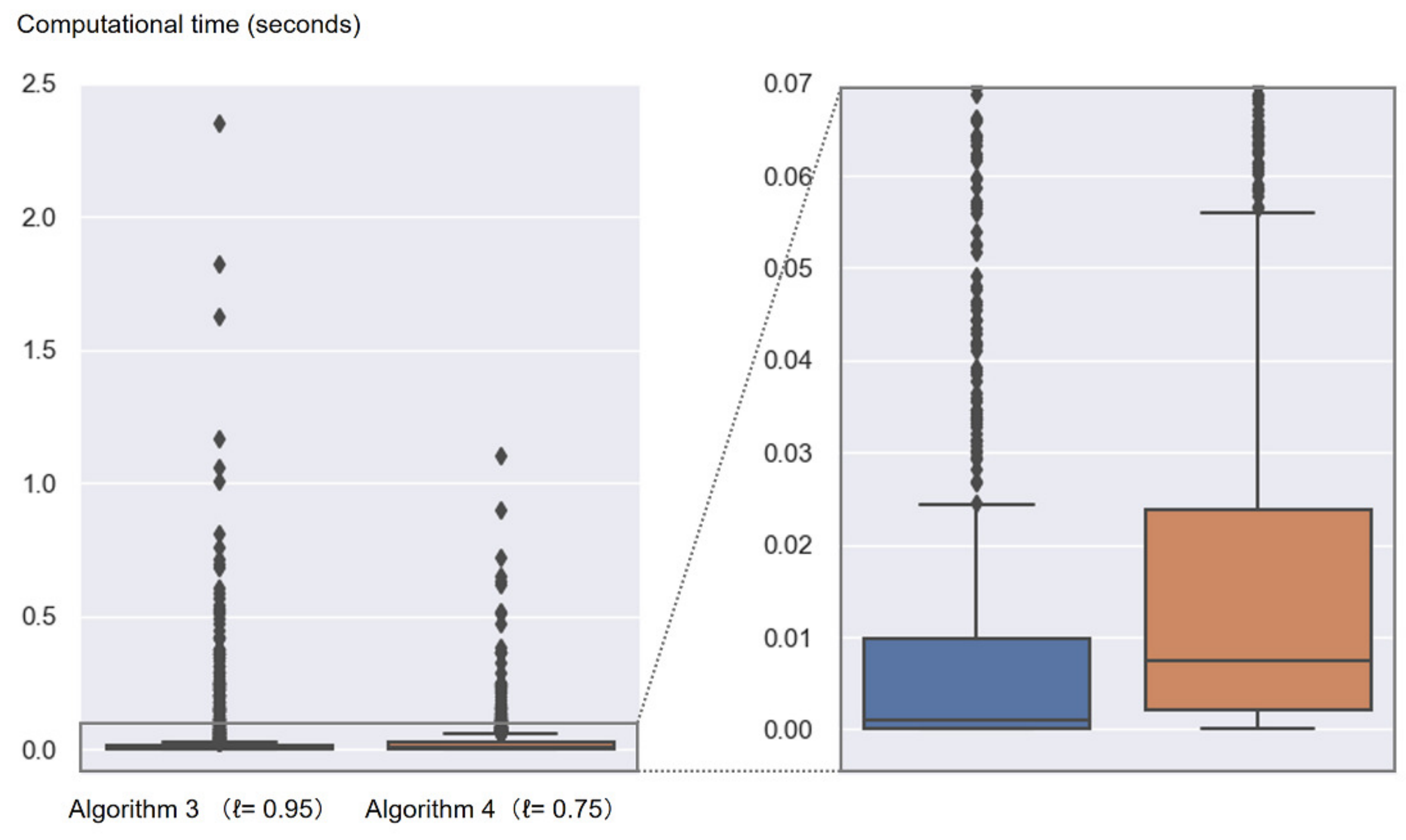}
\caption{Distribution of computational times for Algorithm~\ref{Alg:pruning} and Algorithm~\ref{Alg:dynamic}}
\label{fig:boxplot}
\end{figure}
%
\section{Social Implementation}
%
Our proposed algorithms were installed as core engines 
in the joint transportation matching system \lq\lq TranOpt" by JPR, 
and the service for general users launched on October 21, 2021~\cite{JPR2021} (see Figure~\ref{fig:tranopt}). 
This system enables joint transportation by creating a database of transportation routes for many companies and matching shipper companies across industries using vast amounts of logistics data. 
Users register their company's transportation routes, cargo, and other information in the system, along with their desired conditions for matching. 
The system presents multiple matching candidates to the user, taking into account the desired conditions. 
The user notifies other companies with whom they wishes to share transportation through the system, and joint transportation is executed in a mutually coordinated manner. 
As of October 2022, over 150 companies are using this system. 
We expect dramatic improvements in logistics as this technology will be used widely 
in the future.
\par
Prior to the launch of the service, we conducted a demonstration experiment until the end of August 2021 wherein 100 companies used the system as free-trial users and cooperated in interviews to improve the service. 
During this free trial period, the average occupied vehicle rate of the matching candidates proposed by the system to users was as high as 93\%, 
and users also voiced many expectations of the system. 

\begin{figure}[H]
\centering
\includegraphics[scale=0.12]{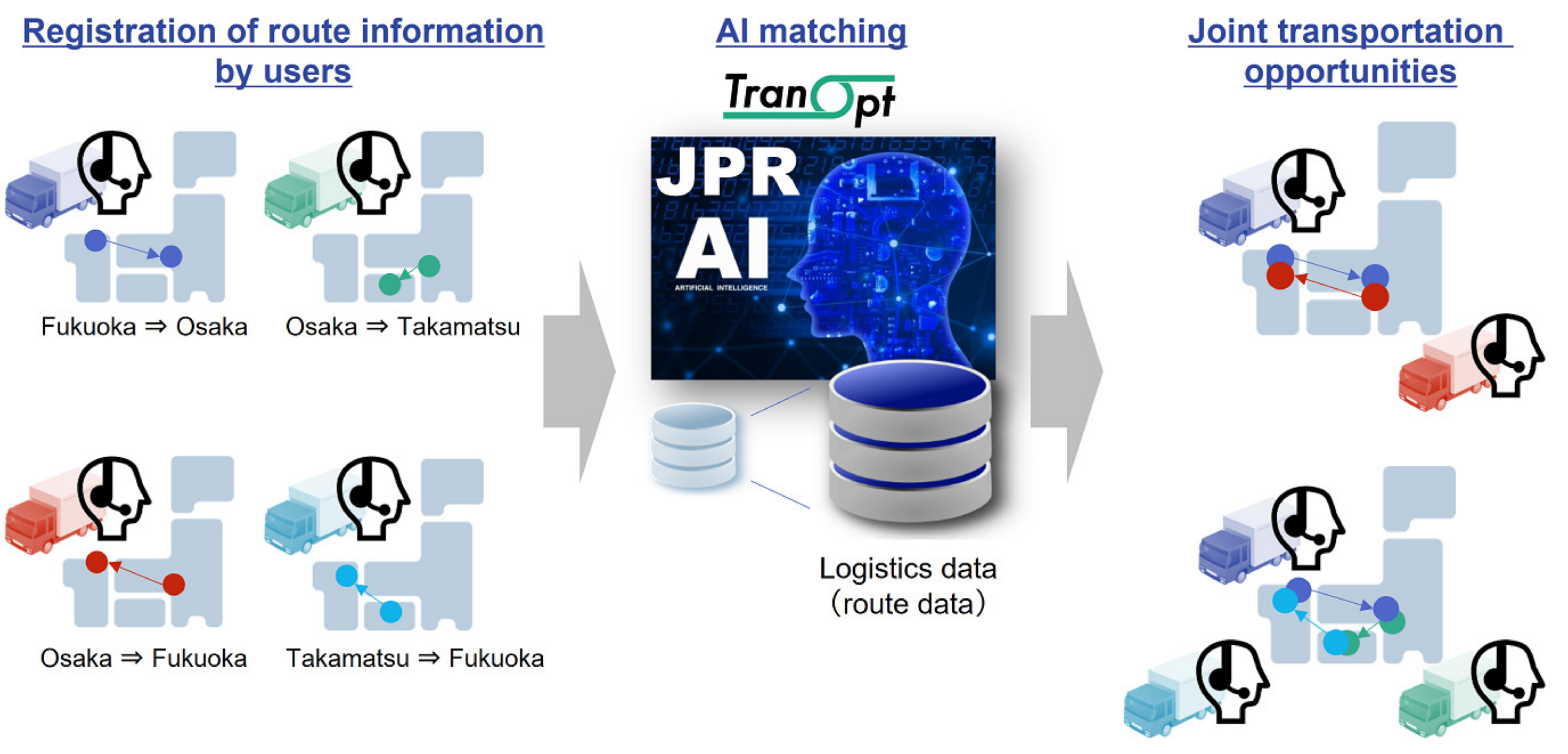}
\caption{Overview of joint transportation matching service}
\label{fig:tranopt}
\end{figure}

Additionally, although this study only deals with triangular transportation, wherein multiple shipments are processed sequentially, it also addresses the high-speed enumeration of mixed transportation, in which loads are mixed and transported simultaneously (see Figure~\ref{fig:mixed}). 
Triangular transportation is expected to improve the occupied vehicle rate, whereas mixed transportation is expected to improve the truck fill rate.

\begin{figure}[h]
\centering
\begin{tabular}{l}
$\longrightarrow$ : Loading Trip \\
$\dashrightarrow$\, : Empty Trip
\end{tabular}
\hspace{5mm}
$\textrm{Reduction Ratio} = \dfrac{\textrm{Length of} \longrightarrow \textrm{when cooperating}}{\textrm{Length of} \longrightarrow \textrm{when not cooperating}}$ \\[3mm]
\includegraphics[scale=0.14]{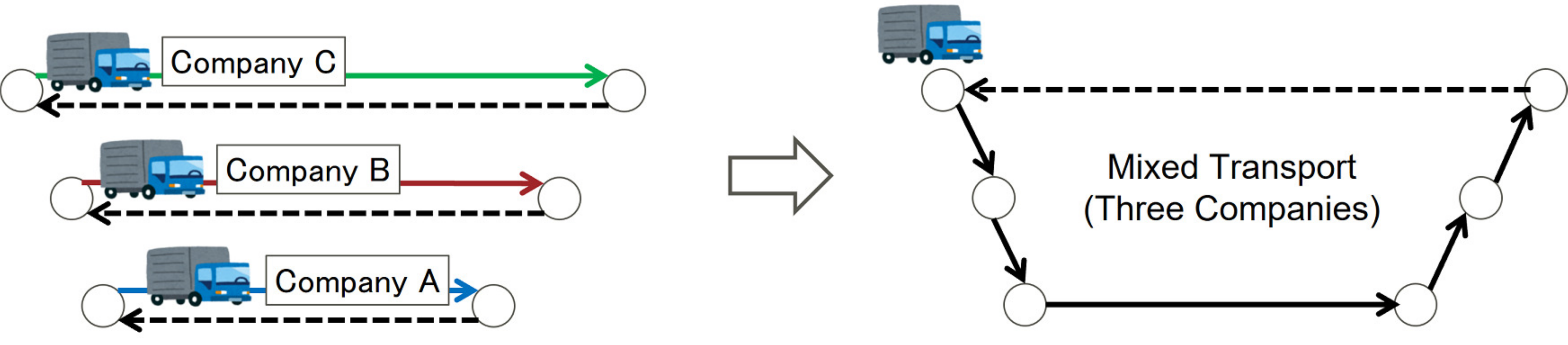}
\caption{Mixed transportation}
\label{fig:mixed}
\end{figure}

In recent years, the importance of the explainability of AI output results has been emphasized. Although this system uses advanced search logic, the output results can be easily explained in the form of \lq\lq all triangular transports with an occupied vehicle ratio of 95\% or above are enumerated" or \lq\lq all mixed transports with a reduction rate of 40\% or less are enumerated." 
In this study, we have just proved the validity of the refinement (i.e., it does not leak any combination that satisfies the condition) for triangular transports. 
Furthermore, we added a mechanism to calculate and display predicted transportation fares for triangular and mixed transportation to users and fair cost sharing among the three companies. Fair cost-sharing is calculated based on the well-known Shapley value~\cite{Shapley1953} in the cooperative game theory. Therefore, in addition to the technologies developed by the authors, the system utilizes the achievements of our predecessors in the OR field.
%
\section{Concluding Remarks}
%
In this study, we have focused on a form of joint transportation called triangular transportation 
and have proposed an algorithm for instantly enumerating combinations with high cooperation effects. 
We Used approximately 17,000 real anonymized transport lane data across Japan. 
We demonstrated that it is thousands of times faster than simple brute-force. 
Based on this enumeration technology, we developed a joint transportation 
matching system. As of October 2022, over 150 companies are using this system. 
We expect dramatic improvements in logistics as this technology will be used widely in the future.
%
\section*{Acknowledgments}
%
The authors would like to thank Prof. Naoyuki Kamiyama, Prof. Katsuki Fujisawa, and Prof. Hidefumi Kawasaki of Kyushu University for their support as advisors in advancing the research and development of the NEDO-funded project \lq\lq 
Cross-industry joint transportation matching service to realize white logistics'' 
(Japan Pallet Rental Corporation). 
\par
Akifumi Kira was supported in part by JSPS KAKENHI Grant Numbers 17K12644 and 
21K11766, Japan.

\begin{thebibliography}{99}
\bibitem{cython}
Behnel,~S., Bradshaw,~R., Citro,~C., Dalcin,~L., Seljebotn,~D.S., \& Smith,~K. (2010). 
Cython: The best of both worlds. 
{\it Computing in Science \& Engineering}, {\bf 13}(2), 31--39.

\bibitem{Creemers2017}
Creemers, S., Woumans, G., Boute, R., \& Beli\"{e}n, J. (2017). 
Tri-Vizor uses an efficient algorithm to identify collaborative shipping opportunities. 
{\it Interfaces}, 
{\bf 47}(3), 244--259.

\bibitem{Cruijssen2007}
Cruijssen, F., Dullaert, W., \& Fleuren, H. (2007). 
Horizontal cooperation in transport and logistics: A literature review. 
{\it Transportation Journal}, 
{\bf 46} (3), 22--39.

\bibitem{ozlem-a}
Ergun, \"{O}., Kuyzu, G., \& Savelsbergh, M. (2007). 
Reducing truckload transportation costs through collaboration.
{\it Transportation Science}, 
{\bf 41}(2), 206--221.

\bibitem{ozlem-b}
Ergun, \"{O}., Kuyzu, G., \& Savelsbergh, M. (2007). 
Shipper collaboration. 
{\it Computers \& Operations Research}, 
{\bf 34}(6), 1551--1560.

\bibitem{Gansterer2018}
Gansterer, M. and Hart, R.F. (2018). 
Collaborative vehicle routing: A survey. 
{\it European Journal of Operational Research}, 
{\bf 268}(1), 1--12.

\bibitem{Ghiani2008}
Ghiani, G., ,Manni, E., \&  Triki, C. (2008)
The lane covering problem with time windows. 
{\it Journal of Discrete Mathematical Sciences and Cryptography}, 
{\bf 11}(1), 67--81.

\bibitem{Guajardo2016}
Guajardo, M., and R\"{o}nnqvist, M. (2016). 
A review on cost allocation methods in collaborative transportation. 
{\it International Transactions in Operational Research}, 
{\bf 23}(3), 371--392. 

\bibitem{JPR2021}
Japan Pallet Rental Corporation and Gunma University (2021). 
Rapid listing of combinations of transport routes with high cooperative effects --- 
Development of joint transport matching technology, 
Joint Press Release (October 21, 2021). \\
\url{https://www.gunma-u.ac.jp/wp-content/uploads/2021/10/Release20211021_EN.pdf}

\bibitem{Karam2021}
Karam, A., Reinau, K.H., \& {\O}stergaard, C.R. (2021). 
Horizontal collaboration in the freight transport sector: barrier and decision-making frameworks. 
{\it European Transport Research Review}, 
{\bf 13}, 1--22.

\bibitem{Kira2021}
Kira, A., Terajima, N., \& Watanabe, Y. (2021).
Transport combination enumeration program, transport combination enumeration method and transport combination enumeration system. 
Japanese Patent Application No. 2021-171440 (October 20, 2021). 

\bibitem{Kuyzu2017}
Kuyzu, G. (2017). 
Lane covering with partner bounds in collaborative truckload transportation procurement. 
{\it Computers \& Operations Research}, 
{\bf 77}, 32--43.

\bibitem{cases}
Ministry of Land, Infrastructure, Transport and Tourism (2022). 
Logistics DX Case Studies for Logistics and Delivery Companies. (in Japanese). \\
\url{https://www.mlit.go.jp/seisakutokatsu/freight/seisakutokatsu_freight_mn1_000018.html}

\bibitem{policy}
Ministry of Land, Infrastructure, Transport and Tourism website. 
Comprehensive Physical Distribution Policy (FY2021--FY2025). (in Japanese), Accessed 2022 Dec 27. \\
\url{https://www.mlit.go.jp/seisakutokatsu/freight/butsuryu03100.html}

\bibitem{Mrabti2022}
Mrabti, N., Hamani, N., \& Delahoche, L. (2022). 
A comprehensive literature review on sustainable horizontal collaboration.
{\it Sustainability}, 
{\bf 14} (18), 11644 (38 pages).

\bibitem{Pan2019}
Pan, S., Trentesaux, D., Ballot, E., \& Huang, G.Q. (2019). 
Horizontal collaborative transport: survey of solutions and practical implementation issues. 
{\it International Journal of Production Research}, 
{\bf 57} (15--16), 5340--5361.

\bibitem{Shapley1953}
Shapley,~L.S. (1953). 
A value for n-person games. 
In Kuhn,~H.W. and Tucker,~A.W. (eds.): 
{\it Contributions to the Theory of Games}, vol. {\bf 2} 
({\it Annals of	mathematics studies}, {\bf 28}), Princeton University
	Press, Princeton, 307--317.

\bibitem{cython-guide}
Smith,~K.W. (2015). 
{\it Cython -- A guide for python programmers}.  
O'Reilly Media, Inc.

\bibitem{Verdonck2013}
Verdonck, L., Caris, A., Ramaekers, K., \& Janssens, G.K. (2013). 
Collaborative logistics from the perspective of road transportation companies. 
{\it Transport Reviews}, 
{\bf 33}(6), 700--719.

\bibitem{heap}
Wiliams,~J.W.J. (1964)  
Algorithm 232: Heapsort. 
{Communications of the ACM}
{\bf 7}(6),
347--348.
\end{thebibliography}


\end{document}